\documentclass[letterpaper, 10 pt, conference]{ieeeconf}
\IEEEoverridecommandlockouts
\usepackage[cmex10]{amsmath}
\usepackage{array}
\usepackage{url}
\usepackage{notations}
\usepackage[utf8]{inputenc}  
\usepackage{citesort}
\usepackage[T1]{fontenc} 
\usepackage{mathtools}
\usepackage{amssymb,amsfonts}
\usepackage{dsfont}
\usepackage{hyperref}
\usepackage{dsfont}
\usepackage{stmaryrd}
\usepackage{bm}

\newenvironment{talign}
 {\let\displaystyle\textstyle\align}
 {\endalign}
\newenvironment{talign*}
 {\let\displaystyle\textstyle\csname align*\endcsname}
 {\endalign}

\def \({\left(}
\def \){\right)}
\def \[{\left[}
\def \]{\right]}

\newcommand{\bu}{{\textbf {u}}}
\newcommand{\bU}{{\textbf {U}}}
\newcommand{\bV}{{\textbf {V}}}
\newcommand{\bF}{{\textbf {F}}}
\newcommand{\bY}{{\textbf {Y}}}

\newcommand{\bW}{{\textbf {W}}}
\newcommand{\bZ}{{\textbf {Z}}}

\newcommand{\bw}{{\textbf {w}}}
\newcommand{\bv}{{\textbf {v}}}

\newcommand{\bX}{{\textbf {X}}}

\newcommand{\be}{\begin{equation}}
\newcommand{\ee}{\end{equation}}
\newcommand{\bea}{\begin{eqnarray}}
\newcommand{\eea}{\end{eqnarray}}

\newcommand\smallO{
  \mathchoice
    {{\scriptstyle\mathcal{O}}}
    {{\scriptstyle\mathcal{O}}}
    {{\scriptscriptstyle\mathcal{O}}}
    {\scalebox{.7}{$\scriptscriptstyle\mathcal{O}$}}
  }

\newtheorem{theorem}{Theorem}[section]
\newtheorem{lemma}[theorem]{\textbf{Lemma}}
\newtheorem{thm}[theorem]{\textbf{Theorem}}
\newtheorem{remark}[theorem]{\textbf{Remark}}
\newtheorem{proposition}[theorem]{\textbf{Proposition}}

\DeclareMathAlphabet{\varmathbb}{U}{bbold}{m}{n}

\newcommand{\EE}{\mathbb{E}}

\newcommand*{\QEDA}{\hfill\ensuremath{\blacksquare}}
\title{\LARGE \bf The Layered Structure of Tensor Estimation \\and its Mutual Information}
\author{Jean Barbier$^{\dagger\Diamond}$, Nicolas Macris$^{\dagger}$ and L\'eo Miolane$^*$\\ \\
$\dagger$ Ecole Polytechnique F\'ed\'erale de Lausanne, Suisse. \\
$\Diamond$ International Center for Theoretical Physics, Trieste, Italy.\\
$*$ INRIA \& Ecole Normale Sup\'erieure, Paris, France.}
\begin{document}
\maketitle
\begin{abstract}
We consider rank-one non-symmetric tensor estimation and derive simple formulas for the mutual information. We start by the order 2 problem, namely matrix factorization. We treat it 
completely in a simpler fashion than previous proofs using a new type of interpolation method developed in \cite{BarbierM17a}. We then show how to harness the 
structure in ``layers'' of tensor estimation in order to obtain a formula for the mutual information for the order 3 problem from the knowledge of the 
formula for the order 2 problem, still using the same kind of interpolation. Our proof technique straightforwardly generalizes and allows to rigorously obtain 
the mutual information at any order in a recursive way. 
\end{abstract}
\section{Introduction}\label{intro}

In the last two decades tensor estimation (also called tensor factorization or decompostion) has found many 
applications in signal processing, high dimensional stastitistics, data mining and machine learning \cite{sidiropoulos2016,cichoki2015,kolda2009}. 
In this contribution we consider simple versions of the problem within a Bayesian framework. 
One observes a {\it noisy version} of an $n$-dimensional, rank-one,
order $p$ tensor $\bU_1\otimes \bU_2 \otimes \dots \otimes \bU_p$ and the goal is to provide an estimate of the $n$-dimensional
random vectors $\bU_i$, $i\!=\!1,\dots, p$ constituting the tensor. 
We consider additive Gaussian noise and in the Bayesian formulation the variance of the noise as well as the priors on the vectors
to be estimated are supposed to be known. A central quantity is the average mutual information, or log-partition function, associated to the Bayes posterior. 
Indeed from this quantity one can typically determine phase transitions as well as  performance measures related to 
minimum mean-square-errors (MMSE). There are very precise conjectures
within this framework that come from analytical computations based on the replica method \cite{mezard1990spin} of
statistical physics and message-passing methods providing the so-called 
approximate message-passing (AMP) algorithm \cite{Donoho10112009,BayatiMontanari10}. These calculations 
have allowed to derive phase diagrams predicting 
stuctural phase transitions inherent 
to the problem, and to compare them to the algorithmic phase transitions \cite{2017arXiv170100858L}. The main finding is that  
there is a region of parameters where AMP performs (in an MMSE sense) as well as the optimal Bayesian estimator,
but there also typically exist large regions of parameters where AMP is sub-optimal or cannot even estimate better than a pure random guess. 
We point out that this phenomenology seems to be quite universal and is found in many other problems related 
to Bayesian inference \cite{mezard2009information}. 

In this contribution we provide a rigorous analysis of the mutual information for rank-one, order $p$ tensor 
estimation in the asymptotic regime $n\!\to\! \infty$. Computing the mutual information (or log-partition sum) a priori involves 
intractable $n$-dimensional integrals or sums. We reduce the problem to low dimensional (typically of order $p$) variational 
expressions which can in principle be solved on a computer. These variational problems also lead to interesting 
questions that are not fully solved, and we provide related conjectures. 

Our method of analysis is based 
on a recently developed {\it adaptive interpolation method} \cite{BarbierM17a} together 
with an inherent {\it layered structure} that underpins the tensor estimation problem: We will relate 
the solution of the order $p\!+\!1$ problem to that of the order $p$ one and provide recursive variational formulas. The case $p\!=\!2$, the so-called ``matrix factorization'' problem, forms the base case and will be presented first as a pedagogical example 
of our interpolation technique. We then explicitly show how to go from $p\!=\!2$ to $p\!=\!3$. The generalization as well as other details 
of our analysis are left to a longer forthcoming contribution. 

Let us briefly say a word on the history of interpolation techniques which are central to this work.
They first originated in the seminal works of Guerra and Toninelli \cite{Guerra-Toninelli-2002,Guerra-2003}
which paved the way towards Talagrand's  proof \cite{Talagrand-annals-2006}
of the Parisi formula \cite{Parisi-1980} for the free eenrgy of the Sherington-Kirkpatrick 
spin glass. Not only these methods have allowed to obtain many more rigorous results on mean-field spin-glasses \cite{Panchenko2013}, 
but remarkably, they have found numerous applications in coding theory, signal processing and theoretical computer 
science problems. So far, Bayesian inference has provided the most fertile ground and replica 
(symmetric) formulas for mutual informations are completely proved in many such cases. A non-exhaustive list of examples is: Coding theory \cite{Giurgiu_SCproof}, random linear estimation \cite{barbier_allerton_RLE,barbier_ieee_replicaCS} and 
matrix factorization \cite{XXT,2016arXiv161103888L,2017arXiv170200473M}. All these works combine the original Guerra-Toninelli interpolation with some {\it other} methodology such as spatial coupling 
\cite{Giurgiu_SCproof,XXT} or the Aizenman-Sims-Starr principle \cite{aizenman2003extended,2017arXiv170200473M}. The present contribution uses a refined form 
of interpolation which is {\it self-contained} and provides what we believe is a 
much simpler and unified approach. 
This approach has also been successfully  
used very recently for non-linear estimation \cite{barbier2017phase}. Finally, as
pointed out above, a special feature of the present problem is the layered structure of tensor estimation 
and we believe that this 
aspect can be leveraged to analyze other relevant multilayered problems. 
\section{Setting and results} 
\label{sec:partI}
\subsection{Non-symmetric tensor estimation}
\subsubsection{Order 2 tensor estimation}
\label{sec:settings}
We use the notation $\bX\!\iid\! P$ to express that the vector (or tensor) $\bX$ has i.i.d.\ components distributed according to $P$. The order $p\!=\!2$ rank-one tensor estimation problem, or matrix factorization, is the task of infering the vectors $\bU\!\in\!\mathbb{R}^{\alpha_u n}$ and $\bV\!\in\!\mathbb{R}^{\alpha_v n}$ (all $\alpha$'s are fixed) from the matrix $\bY\!\in\!\mathbb{R}^{\alpha_u n\times \alpha_v n}$ obtained from the following observation model
\begin{talign}
Y_{ij} = \sqrt{\frac{\lambda}{n}}\,U_iV_j + Z_{ij}, \label{uv}
\end{talign}
for $1\!\le\! i\!\le \alpha_u n$ and $1\!\le\! j\!\le \alpha_v n$. Here $\bZ\!\iid \!{\cal N}(0,1)$ is a Gaussian noise matrix. $\lambda$ is the signal-to-noise ratio and the normalization $1/\sqrt{n}$ makes the estimation problem non trivial. 
We suppose that $\bU\!\iid\! P_u$ and $\bV\!\iid\! P_v$ where the probability distributions $P_u$ and $P_v$ are known by the statistician. We assume that $P_u$ and $P_v$ have a bounded support (this boundedness hypothesis can be relaxed at the end of the proof by approximation arguments, see e.g \cite{barbier2017phase}), and we denote by $\rho_u$ and $\rho_v$ their second moments.

Let $n_u\!=\!\alpha_u n$, $n_v\!=\!\alpha_v n$. We consider a Bayesian setting and associate to the model \eqref{uv} its posterior distribution. The likelihood of the (component-wise conditionally independent) 
observation matrix $\bY$ is 
\begin{talign} \nonumber
P(\bY|\bu, \bv) \!=\! \frac{1}{(2\pi)^{\frac{n_un_v}{2}}} \exp\big\{\!-\frac{1}{2}\!\sum_{i,j=1}^{n_u,n_v}\big(Y_{ij} - \sqrt{\frac{\lambda}{n}}u_i v_j\big)^2\big\},
\end{talign} 
and from the Bayes formula we get the posterior distribution ($\boldsymbol{\Theta}$ is the set of {\it quenched} variables, in this case $\bU, \bV$ and $\bZ$)
\begin{talign}
P(\bu,  \bv|\bY ) = \frac{1}{{\cal Z}(\boldsymbol{\Theta})}P_u(\bu)P_v(\bv) e^{-{\cal H}(\bu,\bv;\boldsymbol{\Theta})} 
\label{post_}
\end{talign}
where we slightly abuse notation by writing $P_u(\bu)\!=\!\prod_{i=1}^{n_u} P_u(u_i)$ and so forth. We employ the vocabulary of statistical physics and call
\begin{talign}
{\cal H}(\bu,\bv;\boldsymbol{\Theta}) 
\equiv \lambda\sum_{i,j=1}^{n_u,n_v}\!\big(\frac{(u_iv_j)^2}{2n} \!-\! \frac{u_iv_jU_iV_j}{n}\!-\!\frac{u_iv_j Z_{ij}}{\sqrt{\lambda n}}\big)
\label{HUV}
\end{talign}
the {\it Hamiltonian} of the model (for obtaining the posterior we replaced $\bY$ by \eqref{uv} and simplified the terms independent of $\bu,\bv$ when normalizing it). The {\it partition function} 
\begin{talign*}
{\cal Z}(\boldsymbol{\Theta})\equiv\int dP_u(\bu) dP_v(\bv) e^{-{\cal H}(\bu,\bv;\boldsymbol{\Theta})}	
\end{talign*}
is the ($\lambda$-dependent) posterior normalization factor. 
%
%
%

Our principal quantity of interest is the average {\it free energy} (the upperscript stands for order $p\!=\!2$)
\begin{talign}
f_n^{(2)}(\lambda)\equiv- \frac{1}{n}\EE\ln {\cal Z}(\boldsymbol{\Theta})\label{f}	
\end{talign}
where $\EE$ always means expectation w.r.t.\ the quenched random variables appearing inside an expression.
It is equal up to an additive constant to the Shannon entropy $H(\bY)$ of $P(\bY)$. This object is related to the mutual information:
\begin{talign}
\frac{1}{n}I(\bU,\bV;\bY) = f_n^{(2)}(\lambda)+\frac{\lambda}{2} \alpha_u\alpha_v\rho_u\rho_v.\label{2MI}
\end{talign}

Its limit $\lim_{n\to\infty}f_n^{(2)}(\lambda)$ contains interesting information such as the location of {\it phase transitions} corresponding to its non analyticity points. Of particular interest is its first (as $\lambda$ is deacreased from infinity) non-analyticity point 
sometimes called the {\it information theoretic threshold} $\lambda_{\rm Opt}$: It is the lowest signal-to-noise ratio such that inference of $(\bU,\bV)$ from $\bY$ is information theoretically ``possible''. Indeed the optimal value of performance measures, such as the minimum mean-square-errors of the vectors, is typically low only above $\lambda_{\rm Opt}$. 
We refer to \cite{barbier_allerton_RLE,barbier_ieee_replicaCS,2016arXiv161103888L,2017arXiv170200473M,2017arXiv170100858L} for more motivations for computing free energies, including algorithmic aspects.
\begin{remark}[Channel universality]\label{rmk:univers}
The Gaussian noise setting \eqref{uv} is actually sufficient to completely characterize the generic model where $\bY$ is observed through a noisy element-wise
(possibly non-linear) output probabilistic channel $P_{\rm out}(Y_{ij}|u_iv_j/\sqrt{n})$. This is made possible by a theorem of 
channel universality \cite{krzakala2016mutual} (conjectured in \cite{lesieur2015mmse} and already proven for community detection in \cite{deshpande2015asymptotic}). The same remark applies to higher order tensor estimation models \cite{lesieur2015mmse,2017arXiv170100858L,2017arXiv170108010L}.
\end{remark}
\subsubsection{Order 3 tensor estimation}
We now observe the order $3$ tensor $\bF$ with entries
\begin{talign}
F_{ijk} = \frac{\sqrt{\lambda}}{n}\,U_iV_jW_k + Z_{ijk} \label{uvw}
\end{talign}
for $1\!\le\! i\!\le\! \alpha_u n$, $1\!\le\! j\!\le \!\alpha_v n$ and $1\!\le\! k\!\le\! \alpha_w n$. The normalization dividing the product of vector is $n^{(p-1)/2}$ for an order $p$ tensor problem; again this normalization makes the problem non-trivial. There is now an additional $\bW\!\in\!\mathbb{R}^{\alpha_w n}$ to infer. It has i.i.d.\ components drawn from the known prior $P_w$ with bounded support and with second moment $\rho_w$. Now $\bZ\iid{\cal N}(0,1)$ is a Gaussian noise tensor.

As for matrix estimation one can define the posterior $P(\bu,\bv,\bw|\bF)$ similarly as \eqref{post_} but with the additional $\bw$ dependence. The associated Hamiltonian ${\cal H}(\bu,\bv,\bw;\boldsymbol{\Theta})$ is equal to 
\begin{talign*}
&\lambda\sum_{i,j,k=1}^{n_u,n_v,n_w}\!\big(\frac{(u_iv_jw_k)^2}{2n^2} \!- \!\frac{u_iv_jw_kU_iV_jW_k}{n^2} \!-\!\frac{u_iv_jw_k Z_{ijk}}{\sqrt{\lambda}n}\big). 
\end{talign*}
Then the average free energy $f_n^{(3)}(\lambda)$ for this model is defined similarly as \eqref{f}, but with the partition function (the normalization of $P(\bu,\bv,\bw|\bF)$) being now ${\cal Z}(\boldsymbol{\Theta}) \!=\! \int dP_u(\bu)dP_v(\bv)dP_w(\bw)\exp(-{\cal H}(\bu,\bv,\bw;\boldsymbol{\Theta}))$. The average free energy is related to the mutual information through
\begin{talign}
\frac{1}{n}I(\bU,\bV,\bW;\bF) = f_n^{(3)}(\lambda)+\frac{\lambda}{2} \alpha_u\alpha_v\alpha_w\rho_u\rho_v\rho_w.\nonumber
\end{talign}

Note the following recursive, or ``layered'', structure linking the order $2$ and $3$ versions of tensor estimation: Conditional on $\bW$, model \eqref{uvw} is an instance of \eqref{uv}. Indeed if one is given $\bW$ and observes \eqref{uvw}
then the inference of $(\bU,\bV)$ from the knowledge of $\bF$ collapses to an order $p\!=\!2$ tensor estimation problem (this remark generalizes to higher orders). This trivial remark is actually essential and is at the root of our recursive proof construction. We will exploit it in order to show that the knowledge of a simple expression for the mutual information of order $p$ tensor estimation can be used to 
obtain one for the order $p\!+\!1$ problem. For pedagogical purpose we prove the result for $p\!=\!2$.
%
\subsection{Variational formulas for the mutual information}
An important role in our proof is played by simple {\it scalar} estimation problems under Gaussian noise. Consider the estimation of the scalar r.v. $X\!\sim\! P_u$ from the observation $Y\!=\!\sqrt{m}\,X\!+\!Z$
%
%
where $Z\!\sim\!{\cal N}(0,1)$ and $m$ plays the role of a signal-to-noise ratio. Then the average free energy for this problem is 
\begin{talign}
\widetilde f_{u}(m)\equiv-\mathbb{E}\ln\int dP_u(x) e^{-m(\frac{x^2}{2} - xX- xZ/\sqrt{m})}. \label{MFf_tensor}
\end{talign}
It is related to minus the expected logarithm of the normalization of the posterior $P(x|Y)$. Define similarly $\widetilde f_{v}$ and $\widetilde f_{w}$ as the r.h.s of \eqref{MFf_tensor} but with $x,X\!\sim\! P_v$ or $P_w$ respectively. Note that by similar computations as in sec.~\ref{sec:dfdt} we have $-2\widetilde f_{u}'(m) = \mathbb{E}\langle x X\rangle_{u,m}\in[0,\rho_u]$, where $\langle -\rangle_{u,m}$ is the expectation w.r.t. the $x$-p.d.f. $\propto dP_u(x) e^{-m(\frac{x^2}{2} - xX- xZ/\sqrt{m})}$.

Define the {\it potential} for matrix estimation as
\begin{talign}
f_{\rm pot}^{(2)}(m_u,m_v;\lambda)&\equiv \frac{\lambda}{2}\alpha_u\alpha_v m_um_v  \nonumber\\
&\ + \alpha_u \widetilde f_{u}(\lambda\alpha_v m_v)+ \alpha_v \widetilde f_{v}(\lambda \alpha_u m_u). \label{f2}
\end{talign}
In the next section we prove using the adaptive interpolation method the following result, first proven in \cite{2017arXiv170200473M} using a more technical strategy based on a rigorous version of the cavity method \cite{mezard1990spin,mezard2009information}, the so-called Aizenman-Sims-Starr scheme \cite{aizenman2003extended}.
In order to state the result we need to introduce the set of critical points of the potential \eqref{f2}:

\vspace{-7pt}
\begin{small}	
	\begin{equation*}
	\Gamma_2(\lambda) \equiv \bigg\{ (m_u,m_v) \in \R_+^2 \, \bigg| \, 
		\begin{array}{rl}
			m_u\!\!\!\! &= -2 {\widetilde{f}}'_u(\lambda \alpha_v m_v)
			\\
			m_v \!\!\!\!&= -2 \widetilde{f}'_v(\lambda \alpha_u m_u)
		\end{array}
	\bigg\}.	
	\end{equation*}	
\end{small}
These equations are known as ``replica-symmetric equations'' in spin glass theory (see \cite{mezard2009information,Talagrand2011spina} for instance) or ``state evolution equations'' in the context of approximate message-passing algorithms \cite{bayati2011dynamics,Montanari-Javanmard}.

\begin{thm}[Free energy of matrix estimation] \label{thm:RS_2}
Fix $\lambda>0$. The average free energy of model \eqref{uv} verifies
\begin{align*}
\lim_{n\to \infty}f_n^{(2)}(\lambda)
&= \inf_{\Gamma_2(\lambda)} f_{\rm pot}^{(2)}(m_u,m_v;\lambda)
\\
&= {\adjustlimits \inf_{m_u}\sup_{m_v}}\, f_{\rm pot}^{(2)}(m_u,m_v;\lambda)
\end{align*}
where this optimization is over $m_u\in[0,\rho_u]$, $m_v\in[0,\rho_v]$.
\end{thm}

Define the potential of the order $3$ problem as
	\begin{align*}
		&f_{\rm pot}^{(3)}(m_u,m_v,m_w;\lambda)\equiv  \lambda \alpha_u\alpha_v\alpha_w m_u m_v m_w \nonumber
		\\
		&\qquad+\alpha_u\widetilde f_u(\lambda \alpha_v\alpha_w m_v m_w) 
		+\alpha_v\widetilde f_v(\lambda \alpha_u\alpha_w m_u m_w) \nonumber
		\\
		&\qquad+\alpha_w\widetilde f_w(\lambda \alpha_u\alpha_v m_u m_v) 
		%
	\end{align*}
	and the corresponding set of critical points:
	
	\vspace{-10pt}
	\begin{small}
		\begin{align*}
			\Gamma_3(\lambda) \!\equiv\! \left\{ \!(m_u,m_v,m_w) \in \R_+^3  \middle| \!\!
				\begin{array}{ll}
					m_u \!\!\!\!\!\!\!&=\! -2 \widetilde{f}'_u(\lambda \alpha_v \alpha_w m_v m_w)
					\\
					m_v \!\!\!\!\!\!\!&= \!-2 \widetilde{f}'_v(\lambda \alpha_u \alpha_w m_u m_w)
					\\
					m_w \!\!\!\!\!\!\!&=\! -2 \widetilde{f}'_w(\lambda \alpha_u \alpha_v m_u m_v)
				\end{array}
			\!\!\!\!\right\}\!.
		\end{align*}
	\end{small}
%
Once Theorem \ref{thm:RS_2} proven, we will use it for obtaining
\begin{thm}[Free energy of tensor estimation] \label{thm:RS_3}
Fix $\lambda>0$. The average free energy of model \eqref{uvw} verifies
\begin{align*}
&\lim_{n\to \infty}f_n^{(3)}(\lambda)
=\inf_{\Gamma_3(\lambda)} f_{\rm pot}^{(3)}(m_u,m_v,m_w;\lambda).
\end{align*}
\end{thm}
\vspace{5pt}
\begin{remark}The fact that the sets $\Gamma_2(\lambda)$ and $\Gamma_3(\lambda)$ are not empty follows from the fact that the functions $-2\widetilde{f}'_u,-2\widetilde{f}'_v,-2\widetilde{f}'_w$ are continuous, bounded and non-negatives (see Lemma 39 in \cite{2017arXiv170200473M}) and from an application of Brouwer's theorem.
\end{remark}

\section{Proofs}
\label{sec:partIII}
The main ingredient of our proof is the adaptive interpolation method recently introduced by two of us in \cite{BarbierM17a}. Note that in contrast with the discrete version of the method 
presented in \cite{BarbierM17a} we will here use it in a continuous form which is even more straightforward for the present problem (yet equivalent). The main difference with the canonical 
interpolation method developed by Guerra and Toninelli in the context of spin glasses \cite{guerra2002thermodynamic,Guerra-2003} is the following: The interpolating 
estimation model that we introduce next is parametrized by
``trial interpolating {\it functions}'' instead of a single trial parameter. These will allow for much more flexibility when choosing the interpolation path, and will actually 
permit us to select an ``optimal'' interpolation path.
\subsection{Initializing the recursion: Proof of Theorem \ref{thm:RS_2}}
\subsubsection{The interpolating model}
%
Let $\epsilon=(\epsilon_u,\epsilon_v)\in[s_n,2s_n]^2$ where $s_n\in(0,1/2]$ is a sequence that goes to $0_+$. Let also $m_{u/v}(s)\in[0,\rho_{u/v}]$ (that can depend on $\epsilon$; also the notation $m_{u/v}$ means $m_{u}$ or $m_{v}$ and similarly for the other quantities). Let the interpolation parameter (or ``time'') $t\!\in\![0,1]$ and the interpolating functions $R_{u/v}(t,\epsilon)\equiv\epsilon_{u/v}+\int_0^t m_{u/v}(s)ds$. Consider the joint estimation of $(\bU,\bV)$ from the three following types of ``time-dependent'' observations
\begin{align}\label{int_model}
\begin{cases}
Y^{(t)}_{ij} &= \sqrt{\frac{\lambda}{n}(1-t)}\,U_iV_j + Z_{ij},\\
Y^{(u,t)}_{i} &= \sqrt{\lambda\alpha_v R_v(t,\epsilon) }\,U_i + Z^{(u)}_i,\\
Y^{(v,t)}_{j} &= \sqrt{\lambda \alpha_u R_u(t,\epsilon)}\,V_j + Z^{(v)}_j,
\end{cases}
\end{align}
for $1\!\le \!i\! \le\! n_u$ and $1\!\le\! j \!\le\! n_v$. Again $\bU\!\iid \!P_u$, $\bV\!\iid\! P_v$ and $\bZ$, $\bZ^{(u)}$, $\bZ^{(v)}\!\iid\!{\cal N}(0,1)$. This model interpolates between the matrix factorization model at $t\!=\!0$ (when there is no ``perturbation'', i.e. $\epsilon=(0,0)$) to a model composed of two independent scalar Gaussian channels (one for $\bU$, one for $\bV$) at $t\!=\!1$. The $\lambda(1\!-\!t)$ appearing in the first set of observations in \eqref{int_model} as well as the interpolating functions all play the role of signal-to-noise ratios, with $t$ giving more and more ``power'' (or weight) to the scalar inference channels when increasing. Here is a crucial and novel ingredient of our interpolation scheme. In classical interpolations, these signal-to-noise ratios (snr) would all take a trivial form (i.e.\ would be linear in $t$) but here, the additional degree of freedom gained from the non-trivial dependency in $t$ of the two latter snr through the introduction of the interpolating functions will be essential.

Let us define the following {\it interpolating Hamiltonian} ${\cal H}_{t,\epsilon}\!=\!{\cal H}_{t,\epsilon}(\bu,\bv)$ (from now on we do not indicate anymore the dependence w.r.t. quenched variables to ease the notations) associated with \eqref{int_model}:

\vspace{-7pt}
\begin{talign}
&{\cal H}_{t,\epsilon}\equiv \lambda \alpha_v R_v(t,\epsilon)  \sum_{i=1}^{n_u}\!\big(\frac{u_i^2}{2} - u_iU_i-\frac{u_i Z_{i}^{(u)}}{\sqrt{\lambda\alpha_v R_v(t,\epsilon)}}\big) \nonumber \\
&+\lambda \alpha_u R_u(t,\epsilon)  \sum_{j=1}^{n_v}\!\big(\frac{v_j^2}{2} - v_jV_j-\frac{v_j Z_{j}^{(v)}}{\sqrt{\lambda\alpha_u R_u(t,\epsilon)}}\big) \nonumber\\
&+\lambda(1\!-\!t)\sum_{i,j=1}^{n_u,n_v}\!\big(\frac{(u_iv_j)^2}{2n} - \frac{u_iv_jU_iV_j}{n}-\frac{u_iv_j Z_{ij}}{\sqrt{\lambda(1-t) n}}\big).\label{Ht}
\end{talign}
%
%
We note that for $t\!=\!0$ and $\epsilon=(0,0)$ so that both $R_{u/v}$ cancel, this Hamiltonian is \eqref{HUV}. 
This Hamiltonian relates to the $t$-dependent posterior of the interpolating model through
\begin{talign}
P_{t,\epsilon}(\bu,\bv)=\frac{1}{{\cal Z}_{t,\epsilon}}P_u(\bu)P_v(\bv)e^{-{\cal H}_{t,\epsilon}} \label{post_t}
\end{talign}
where ${\cal Z}_{t,\epsilon}$ is the normalization. To \eqref{post_t} is associated a {\it Gibbs bracket} $\langle - \rangle_{t,\epsilon}$ defined as
$\langle A\rangle_{t,\epsilon} = \int  dP_{t,\epsilon}(\bu,\bv) A(\bu,\bv)$.
%
Moreover the interpolating free energy is
\begin{talign*}
f_{n,\epsilon}(t)\equiv-\frac1n\EE\ln {\cal Z}_{t,\epsilon} =-\frac1n\EE\ln\int dP_u(\bu)dP_v(\bv) e^{-{\cal H}_{t,\epsilon}}
\end{talign*}
where here $\mathbb{E}$ is the expectation w.r.t. $\bU, \bV,\bZ,\bZ^{(u)}$ and $\bZ^{(v)}$.

\subsubsection{Overlap concentration} \label{overlap_conc}
Let us define the {\it overlaps} 
\begin{talign*}
Q_u \equiv \frac{1}{n_u}\sum_{i=1}^{n_u} u_i U_i \quad \text{and}\quad Q_v \equiv \frac{1}{n_v}\sum_{j=1}^{n_v} v_j V_j	
\end{talign*}
where $(\bu,\bv)$ are jointly drawn from the posterior \eqref{post_t}. The next lemma plays a key role in our proof. Essentially it states that the overlaps concentrate around their mean, a behavior called ``replica symmetric'' in statistical physics.
Similar results have been obtained in the context of the analysis of spin glasses \cite{Talagrand2011spina}. Here we use a formulation taylored to Bayesian 
inference problems as developed in the context of LDPC codes, linear estimation and Nishimori symmetric spin glasses \cite{GKSmacris2007,KoradaMacris_CDMA,korada2009exact}. 

In order to prove this concentration the precense of the ``small'' perturbation $\epsilon=(\epsilon_u,\epsilon_v)$ is crucial. It can be interpreted as having extra observations coming from Gaussian side-channels $\widehat Y_i^{(u)}\! =\! \sqrt{\epsilon_u}\, U_i \!+\! \widehat{Z}_i^{(u)}$ and similarly for $\bV$. This perturbation induces only a small change in the free energy, namely of the order of the $\epsilon$'s. Indeed, a simple computation 
gives that at $t=0$,
$
|\partial_{\epsilon_{u/v}} f_{n,\epsilon}(t=0)| =  \frac12| \EE \langle Q_{u/v} \rangle_{0,\epsilon}|.
$
The overlaps are bounded for priors $P_u$, $P_v$ of bounded support. This implies that for priors with bounded supports, we have for all $\epsilon\in[s_n,2s_n]^2$ and $t=0$ that
\begin{align}
|f_{n,\epsilon}(0) - f_{n,\epsilon=(0,0)}(0) | \leq C\,s_n \to 0_+	\label{25}
\end{align}
for some constant $C$ depending only on the priors $P_{u}$ and $P_v$ (and recall that $f_{n,\epsilon=(0,0)}(0)=f_n^{(2)}(\lambda)$)
This small perturbation forces the overlaps to concentrate.
\begin{lemma}[Overlap concentration] \label{concentration}
Assume that for any $t \in (0,1)$ the map $\epsilon=(\epsilon_u,\epsilon_v) \in [s_n,2s_n]^2 \mapsto R(t,\epsilon)=(R_u(t,\epsilon),R_v(t,\epsilon))$ is a $\mathcal{C}^1$ diffeomorphism with Jacobian determinant greater or equal to $1$. Then one can find a sequence $s_n$ going to $0$ slowly enough such that there exist positive constants $C$ and $\gamma$ that only depend on the support and moments of $P_u$ and $P_v$ and on $\alpha$, and such that:
 \begin{align*}
\textstyle \frac{1}{s_n^2}\int_{[s_n,2s_n]^2} d\epsilon\int_0^1 dt\, \mathbb{E}\big\langle \big(Q_u - \mathbb{E}\langle Q_u\rangle_{t, \epsilon}\big)^2 \big\rangle_{t, \epsilon}  \leq C n^{-\gamma}
\end{align*}
and similarly for $Q_v$.
\end{lemma}

We refer to \cite{BarbierM17a,barbier2017phase} for more details where the method used to show the overlap concentration has been streamlined. Note that the method is based on the concentration of the free energy around its average w.r.t. the quenched variables. In \cite{BarbierM17a} this concentration is proven for the problem of symmetric rank-one matrix factorization but the proof straightforwardly generalizes to non-symmetric tensors.
\subsubsection{Adaptive interpolation} \label{stoInt_uv}
We now have all the necessary ingredients to prove Theorem \ref{thm:RS_2} using the adaptive interpolation method. The first step is to notice, using in particular identity \eqref{25}, that $f_{n,\epsilon}(t)$ verifies
\begin{align}
\begin{cases}\label{23}
f_{n,\epsilon}(0)&= f_n^{(2)}(\lambda) + {\cal O}(s_n),\\
f_{n,\epsilon}(1)&=\alpha_u \widetilde f_{u}(\lambda\alpha_v R_v(1))+ \alpha_v \widetilde f_{v}(\lambda \alpha_u R_{u}(1)).
\end{cases}	
\end{align}
So at $t\!=\!0$ one (almost) recovers 
the average free energy \eqref{f} of the original model, while at $t\!=\!1$ appear two terms of the potential \eqref{f2}. This is the reason for the introduction of the scalar channels in \eqref{int_model}. In order to compare $f_n^{(2)}(\lambda)$ with the 
potential we use the fundamental theorem of calculus 
\begin{talign*}
f_{n,\epsilon}(0)=f_{n,\epsilon}(1)-\int_0^1 f_{n,\epsilon}'(t) dt.	
\end{talign*}
It is thus natural to compute (see sec.~\ref{sec:dfdt} for the proof)
\begin{talign}
	f_{n,\epsilon}'(t) =&-\frac{\lambda }{2}\alpha_u\alpha_v\big\{m_u(t)m_v(t) \nonumber\\
&-\EE\big\langle (Q_u - m_u(t))(Q_v-m_v(t))\big\rangle_{t,\epsilon}\big\}. \label{dfdt}
\end{talign}
Replacing \eqref{23} and \eqref{dfdt} in the fundamental theorem of calculus and using Lemma~\ref{concentration} for $Q_u$, $Q_v$ together with the Cauchy-Schwarz inequality leads that $f_n^{(2)}(\lambda)$ is equal to
\begin{talign}
&\frac{1}{s_n^2}\!\int d\epsilon\big[\alpha_u \widetilde f_{u}(\lambda\alpha_v \int_0^1 m_v(t)dt)\!+\! \alpha_v \widetilde f_{v}(\lambda \alpha_u  \int_0^1 m_u(t)dt)\nonumber\\
& +\frac{\lambda}{2} \alpha_u\alpha_v\big\{{\textstyle \int_0^1 dt\, m_u(t)m_v(t)}\label{14}\\
&-{\textstyle \int_0^1 dt(\EE\langle Q_u\rangle_{t,\epsilon} \!-\! m_u(t))(\EE\langle Q_v\rangle_{t,\epsilon}\!-\!m_v(t))\big\}}\big]+\smallO_n(1),\nonumber
\end{talign}
where $\smallO_n(1)$ denotes a quantity that goes to $0$ as $n\to \infty$, uniformly in $t, m_{u},m_v,\epsilon$. To obtain this last identity we also used the continuity and boundedness of $\widetilde f_{u}$ and $\widetilde f_{v}$ (see e.g. \cite{barbier2017phase} or sec. 7.1 in \cite{2017arXiv170200473M}). We are now in position to provide the core identity of our proof scheme:
\begin{lemma}[Sum rule]\label{sumRule} 
Assume $\epsilon\mapsto R(t,\epsilon)$ satisfies the hypotheses of Lemma~\ref{concentration}, and choose $s_n\to 0_+$ according to this lemma. Assume that for all $t \in [0,1]$ and $\epsilon \in [s_n,2s_n]^2$ we have $m_v(t,\epsilon) = \mathbb{E}\langle Q_v \rangle_{t,\epsilon}$. Then:
\begin{talign*}
&f_n^{(2)}(\lambda)= \smallO_n(1)+\frac{1}{s_n^2}\!\int d\epsilon \big[\frac{\lambda}{2} \alpha_u\alpha_v \int_0^1 dt\, m_u(t,\epsilon)m_v(t,\epsilon)\\
&\quad+\alpha_u \widetilde f_{u}(\lambda\alpha_v \int_0^1 m_v(t,\epsilon)dt)+ \alpha_v \widetilde f_{v}(\lambda \alpha_u  \int_0^1 m_u(t,\epsilon)dt)\big]
\end{talign*}	
where $\smallO_n(1)$ is uniform in $t, m_{u},m_v,\epsilon$.
\end{lemma}

From this we can derive in a unified way matching bounds. But first we emphasize on a crucial and novel point of the adaptive interpolation method which makes it quite powerful: In previous interpolations, the remainder (i.e. the last term in \eqref{14}) always remains and if by luck it has an obvious sign, then comparisons between the left and right hand sides of identities like \eqref{14} may eventually lead to a (single-sided) bound. But with our method the remainder can be directly {\it canceled}, which allows to obtain much stronger results irrespective of the remainder's sign as we show now.
\subsubsection{Matching bounds} \label{sec:matrix_bounds}
Similar bounds can be found in \cite{barbier2017phase}, to which we will refer when needed for more details. 

$\bullet$ {\it Upper bound:} Let $m_u(t)=m_u\in[0,\rho_u]$ be a constant. We then fix $R = (R_u,R_v)$ as the solution $R(t,\epsilon)=(\epsilon_u+m_u t,\epsilon_v+\int_0^t m_v(s,\epsilon)ds)$ to the first order differential equation: $\partial_t R_u(t) = F_u$, $\partial_t R_v(t) = F_v(t,R(t))$, and $R(0) = \epsilon$, with $F_u\equiv m_u$, $F_v(t,R(t))\equiv \mathbb{E}\langle Q_v \rangle_{t,\epsilon}$ which takes values in $[0,\rho_v]$. One can check (see \cite{barbier2017phase}) that this ODE satisfies the hypotheses of the Cauchy-Lipschitz theorem. As $F=(F_u,F_v)$ (which also depends on $n$) is continuous and admits continuous partial derivatives, $R(t,\epsilon)$ is ${\cal C}^1$ (in both arguments). By the Liouville formula, the Jacobian determinant $J_{n,\epsilon}(t)$ of $\epsilon \mapsto R(t,
\epsilon)$ satisfies $J_{n,\epsilon}(t) = \exp \{\int_0^t\partial_{R_v} F_v(s,R(s,\epsilon))ds\} \geq 1$; indeed, $\partial_{R_v} F_v\ge 0$, see Prop.~6 of \cite{barbier2017phase}. Also, as this Jacobian never cancels, and as $\epsilon \mapsto R(t,
\epsilon)$ is injective (by unicity of $R(t,
\epsilon)$), it is a diffeomorphism by the inversion theorem. Recall \eqref{f2}. Then Lemma.~\ref{sumRule} implies: 
\begin{align*}
f_n^{(2)}(\lambda) \!=\! {\textstyle \frac{1}{s_n^{2}}\!  \int_{[s_n,2s_n]^2} d\epsilon \, f_{\rm pot}^{(2)}(m_u,\int_0^1 m_v(t,\epsilon)dt;\lambda)\!+\! \smallO_n(1)	}.
\end{align*}
Thus $\limsup_{n\to\infty}f_n^{(2)}(\lambda)\le { \inf_{m_u}\sup_{m_v}}f_{\rm pot}^{(2)}(m_u,m_v;\lambda)$
where the optimization is over $m_u\in[0,\rho_u]$, $m_v\in[0,\rho_v]$.

$\bullet$ {\it Lower bound:} Fix $R$ as the solution $R(t,\epsilon)=(\epsilon_u+\int_0^t m_u(s,\epsilon)ds,\epsilon_v+\int_0^t m_v(s,\epsilon)ds)$ to the following Cauchy problem: $\partial_t R_u(t) = F_u(t,R(t))\equiv -2\widetilde f_u'(\lambda\alpha_v\mathbb{E}\langle Q_v\rangle_{t,\epsilon})$ (recall \eqref{MFf_tensor}) and $\partial_t R_v(t) = F_v(t,R(t)) \equiv \mathbb{E}\langle Q_v\rangle_{t,\epsilon}$ with $R(0) = \epsilon$. We denote this equation $\partial_t R(t) = F(t,R(t))$. The solutions verify $m_u(t,\epsilon) \in [0,\rho_u]$ and $m_v(t,\epsilon)=\mathbb{E}\langle Q_v\rangle_{t,\epsilon} \in [0,\rho_v]$. It is possible to verify (see the details in a similar case in \cite{barbier2017phase}) that $F(t,R)$ is a bounded $\mathcal{C}^1$ function of $R$, and thus the Cauchy-Lipschitz theorem implies that $R(t,\epsilon)$ is a $\mathcal{C}^1$ function of both $t$ and $\epsilon$. The Liouville formula for the Jacobian determinant $J_{n,\epsilon}(t)$ of the map $\epsilon \mapsto R(t,\epsilon)$ yields $J_{n,\epsilon}(t) = \exp\{\int_0^t\partial_{R_u} F_u(s,R(s,\epsilon))ds+\int_0^t\partial_{R_v} F_v(s,R(s,\epsilon))ds\} \geq 1$ as both partial derivatives (in the exponential) are non-negative for all $s \in (0,1)$ (see again \cite{barbier2017phase}). By the same arguments as in the previous bound, for any $t$, the map $\epsilon \mapsto R(t,\epsilon)$ a ${\cal C}^1$ diffeomorphism. All hypotheses of Lemma.~\ref{sumRule} are verified. It leads to:
\begin{talign*}
&f_n^{(2)}(\lambda)= \smallO_n(1)+\frac{1}{s_n^2}\!\int d\epsilon\big[\frac{\lambda}{2} \alpha_u\alpha_v \int_0^1 dt\, m_u(t,\epsilon)m_v(t,\epsilon)\nonumber\\
&+\alpha_u \widetilde f_{u}(\lambda\alpha_v \int_0^1 m_v(t,\epsilon)dt)+ \alpha_v \widetilde f_{v}(\lambda \alpha_u  \int_0^1 m_u(t,\epsilon)dt)\big].
\end{talign*}
Both $\widetilde f_{u}$ and $\widetilde f_{v}$ are concave (it is simple to show, see e.g. \cite{barbier2017phase}). Jensen's inequality thus yields (and recalling \eqref{f2})
\begin{eqnarray*}
\textstyle f_n^{(2)}(\lambda)\geq \frac{1}{s_n^{2}} \int d\epsilon \int_0^1  dt \, f_{\rm pot}^{(2)}(m_u(t,\epsilon),m_v(t,\epsilon);\lambda) + \smallO_n(1),
\end{eqnarray*}
($\epsilon$ is integrated over $[s_n,2s_n]^2$). Notice now that 
\begin{align*}
f_{\rm pot}^{(2)}(m_u(t,\epsilon),m_v(t,\epsilon);\lambda)=\sup_{m_v\in[0,\rho_v]}f_{\rm pot}^{(2)}(m_u(t,\epsilon),m_v;\lambda).	
\end{align*}
Indeed, $g_{m_u} : m_v \mapsto f_{\rm pot}^{(2)}(m_u,m_v;\lambda)$ is concave (by concavity of $\widetilde f_{v}$), with derivative $g_{m_u}'(m_v) =\frac{\lambda}{2}\alpha_u\alpha_v[m_u + 2\widetilde f_u'(\lambda \alpha_v m_v)]$. By definition of the solution $R(t,\epsilon)$ of the ODE, $g_{m_u(t,\epsilon)}'(m_v(t,\epsilon)) = 0$ for any $(t,\epsilon)$, so by concavity $g_{m_u(t,\epsilon)}$ reaches its maximum at $m_v(t,\epsilon)$. Therefore,
\begin{talign*}
f_n^{(2)}\!(\lambda)&\geq \frac{1}{s_n^{2}} \int d\epsilon \int_0^1  dt  \sup_{m_v} f_{\rm pot}^{(2)}(m_u(t,\epsilon),m_v;\lambda) \!+\! \smallO_n(1)\\
&\geq \inf_{m_u}  \sup_{m_v} f_{\rm pot}^{(2)}(m_u,m_v;\lambda) \!+\! \smallO_n(1).
\end{talign*}
Thus $\liminf_{n\to\infty}f_n^{(2)}(\lambda)\ge { \inf_{m_u}\sup_{m_v}}f_{\rm pot}^{(2)}(m_u,m_v;\lambda)$, which ends the proof of the second equality of Theorem \ref{thm:RS_2}. The first equality follows from Lemma \ref{lem:sup_inf} in appendix. $\QEDA$
\subsubsection{Proof of \eqref{dfdt}}\label{sec:dfdt}
Let us show how the derivative of the interpolating free energy is obtained. Is is given by
\begin{talign*}
	f_{n,\epsilon}'&(t)=\frac{1}{n}\EE\big\langle \frac{d{\cal H}_t}{dt}\big \rangle_t \\
&=\frac{\lambda \alpha_vm_v(t)}{n}\EE\big\langle\sum_{i=1}^{n_u}\big(\frac{u_i^2}{2}-u_iU_i-\frac{u_iZ_i^{(u)}}{2\sqrt{\lambda \alpha_v R_v(t)}}\big)\big \rangle_t\nonumber\\
&\quad+\frac{\lambda \alpha_um_u(t)}{n}\EE\big\langle\sum_{j=1}^{n_v}\big(\frac{v_j^2}{2}-v_jV_j-\frac{v_jZ_j^{(v)}}{2\sqrt{\lambda \alpha_uR_u(t)}}\big)\big \rangle_t\nonumber\\
&\quad-\frac{\lambda}{n} \EE\big\langle\sum_{i,j=1}^{n_u,n_v}\big(\frac{(u_iv_j)^2}{2n}-\frac{u_iv_jU_iV_j}{n}-\frac{u_iv_jZ_{ij}}{2\sqrt{\lambda (1-t)n}}\big)\big \rangle_t.\nonumber
\end{talign*}
Let $(\bu',\bv')$ be jointly drawn from the posterior \eqref{post_t} and this independently from $(\bu,\bv)$, itself also drawn from the same posterior. We now integrate by part the Gaussian noise variables using the elementary formula $\EE[Za(Z)]\!=\!\EE[a'(Z)]$ for $Z\!\sim\!{\cal N}(0,1)$ and for continuously differentiable $a$ such that these expectations are well-defined. This leads to
\begin{talign*}
	f_{n,\epsilon}'(t)&=\frac{\lambda \alpha_vm_v(t)}{n}\EE\big\langle\sum_{i=1}^{n_u}\big(-u_iU_i+\frac{u_iu_i'}{2}\big)\big \rangle_t\nonumber\\
&\quad+\frac{\lambda \alpha_um_u(t)}{n}\EE\big\langle\sum_{j=1}^{n_v}\big(-v_jV_j+\frac{v_jv_j'}{2}\big)\big \rangle_t\nonumber\\
&\quad-\frac{\lambda}{n} \EE\big\langle\sum_{i,j=1}^{n_u,n_v}\big(-\frac{u_iv_jU_iV_j}{n}+\frac{u_iu_i'v_jv_j'}{2n}\big)\big \rangle_t.
\end{talign*}
We now use the following identities $\EE\langle u_iU_i \rangle_t=\EE\langle u_iu_i' \rangle_t$ and $\EE\langle v_jV_j \rangle_t=\EE\langle v_jv_j' \rangle_t$. 
These follow directly from the following identity
which is nothing more than a direct consequence of Bayes formula (see \cite{2017arXiv170200473M,barbier_ieee_replicaCS} for a proof): $\EE\langle g(\bu, \bv,\bU,\bV) \rangle_t = \EE\langle g(\bu, \bv,\bu',\bv') \rangle_t$
%
%
for any continuous bounded function $g$. Thus
\begin{talign*}
f_{n,\epsilon}'&(t)=-\frac{\lambda \alpha_vm_v(t)}{n}\EE\big\langle\sum_{i=1}^{n_u}\frac{u_iU_i}{2}\big \rangle_t\nonumber\\
&\,-\frac{\lambda \alpha_um_u(t)}{n}\EE\big\langle\!\sum_{j=1}^{n_v}\frac{v_jV_j}{2}\big \rangle_t+\frac{\lambda}{n} \EE\big\langle\sum_{i,j=1}^{n_u,n_v}\frac{u_iv_jU_iV_j}{2n}\big \rangle_t\nonumber\\
&\ \ \ \, =-\frac{\lambda \alpha_u \alpha_v}{2}\EE\big\langle m_v(t) Q_u+m_u(t) Q_v- Q_uQ_v\big\rangle_t
\end{talign*}
which is \eqref{dfdt}. $\QEDA$
\subsection{From $p=2$ to $p=3$: Proof of Theorem \ref{thm:RS_3}}
Let us now prove the second theorem using our previous findings, using again the adaptive interpolation method. We will start by proving an alternative version of the limit of the free energy, using an 
auxiliary potential:
\begin{talign}
&f_{\rm aux}^{(3)}(m_w,m_{uv};\lambda) \equiv  \frac{\lambda }{2}\alpha_u\alpha_v\alpha_wm_{uv} m_w \label{f3aux}\\
&+\alpha_w\widetilde f_w(\lambda \alpha_u\alpha_v m_{uv}) + {\adjustlimits \inf_{m_{u}}\sup_{m_{v}}}\, f_{\rm pot}^{(2)}(m_u,m_v;\lambda \alpha_wm_w)\nonumber		
\end{talign}
where the optimization is over $m_{u/v}\in[0,\rho_{u/v}]$.

\begin{proposition}[Auxiliary free energy formula] \label{inf_sup_inf_sup}
Fix $\lambda>0$. The average free energy of model \eqref{uvw} verifies
\begin{align*}
&\lim_{n\to \infty}\!f_n^{(3)}(\lambda)
\!=\!{\adjustlimits \inf_{m_w}\sup_{m_{uv}} }\, f_{\rm aux}^{(3)}(m_w,m_{uv};\lambda)
\end{align*}
with optimization over $m_{w}\in[0,\rho_w]$ and $m_{uv}\in[0,\rho_u\rho_v]$.
\end{proposition}

Once Proposition \ref{inf_sup_inf_sup} will be proved, Theorem \ref{thm:RS_3} will simply follow from Lemma \ref{lem:SE_equiv} presented in appendix ($\widetilde{f}_u$, $\widetilde{f}_v$ and $\widetilde{f}_w$ are indeed strictly concave, differentiable, Lipschitz, non-increasing functions over $\R_+$ by Lemma 39 from \cite{2017arXiv170200473M}).
\subsubsection{The ``layered'' interpolating model}
Similarly as before $t\!\in\![0,1]$, $\epsilon=(\epsilon_w,\epsilon_{uv})\in[s_n,2s_n]^2$ with $s_n\in(0,1/2]$ going to $0_+$ and the interpolating functions $R_{w/uv}(t,\epsilon)\equiv\epsilon_{w/uv}+\int_0^t m_{w/uv}(s)ds$ with $m_{w}(s)\!\in\![0,\rho_w]$ and $m_{uv}(s)\!\in\![0,\rho_u\rho_v]$. Consider this time the following observation model:
\begin{align*}
\begin{cases}
F^{(t)}_{ijk} &= \frac{\sqrt{\lambda(1-t)}}{n}\,U_iV_jW_k + Z_{ijk},\\
Y^{(uv,t)}_{ij} &= \sqrt{\frac{\lambda}{n}\alpha_w R_w(t)}\,U_iV_j + Z^{(uv)}_{ij},\\
Y^{(w,t)}_{k} &= \sqrt{\lambda \alpha_u\alpha_vR_{uv}(t)}\,W_k + Z^{(w)}_k,
\end{cases}
\end{align*}
for $1\!\le \!i \!\le\! n_u$, $1\!\le \!j \!\le\! n_v$ and $1\!\le\! k \!\le\! n_w$. Here $\bU\!\iid\! P_u$, $\bV\!\iid\! P_v$, $\bW\!\iid\! P_w$ and $\bZ$, $\bZ^{(uv)}$, $\bZ^{(w)}\!\iid\!{\cal N}(0,1)$. This model interpolates 
between an order $p\!+\!1\!=\!3$ tensor estimation at $t\!=\!0$ and $\epsilon=(0,0)$, to a model combined of a scalar estimation problem over $\bW$ under Gaussian noise and an order $p\!=\!2$ tensor joint estimation problem over $(\bU,\bV)$ at $t\!=\!1$. This model is ``layered'' in the 
sense that one decoupled scalar estimation problem is considered in addition of the order $p\!=\!2$ joint estimation problem that has {\it already} been treated analytically, see Theorem \ref{thm:RS_2}.

As previously, we associate to this model its posterior distribution $P_{t,\epsilon}(\bu,\bv,\bw)\!=\!{\cal Z}_{t,\epsilon}^{-1}P_u(\bu)P_v(\bv)P_w(\bw)\exp(-{\cal H}_{t,\epsilon})$ 
where the interpolating Hamiltonian ${\cal H}_{t,\epsilon}\!=\!{\cal H}_{t,\epsilon}(\bu,\bv,\bw)$ (again quenched variables are not indicated explicitly) is
\begin{talign*}
&{\cal H}_{t,\epsilon} = \lambda(1-t)\sum_{i,j,k=1}^{n_u,n_v,n_w}\big(\frac{(u_iv_jw_k)^2}{2n^2} - \frac{u_iv_jw_kU_iV_jW_k}{n^2} \nonumber\\
&\qquad\qquad\qquad\qquad\qquad\qquad-\frac{u_iv_jw_k Z_{ijk}}{\sqrt{\lambda n^2(1-t)}}\big) 
\nonumber \\ & + \lambda\alpha_w R_w(t)\sum_{i,j=1}^{n_u,n_v}\big(\frac{(u_iv_j)^2}{2n} - \frac{u_iv_jU_iV_j}{n} 
-\frac{u_iv_j Z_{ij}^{(uv)}}{\sqrt{\lambda n \alpha_w R_w(t) }}\big)\nonumber\\
&+\lambda \alpha_u\alpha_v  R_{uv}(t) \sum_{k=1}^{n_w}\big(\frac{w_k^2}{2} - w_kW_k-\frac{w_k Z_{k}^{(w)}}{\sqrt{\lambda\alpha_u\alpha_v R_{uv}(t)}}\big).
\end{talign*}
%
%
The Gibbs bracket $\langle - \rangle_{t,\epsilon}$ is, as before, the expectation w.r.t.\ this posterior. Finally the interpolating free energy is 
\begin{talign*}
f_{n,\epsilon}(t) =-\frac1n\mathbb{E}\ln\int dP_u(\bu)dP_v(\bv)dP_w(\bw) e^{-{\cal H}_t}.
\end{talign*}
\subsubsection{Adaptive interpolation}
The steps that we follow now are similar to sec. \ref{stoInt_uv}. The free energy $f_{n,\epsilon}(t)$ verifies, using identity \eqref{25} and $f_{n,\epsilon=(0,0)}(0)=f_n^{(3)}(\lambda)$,
\begin{align*}
\begin{cases}
f_{n,\epsilon}(0)= f_n^{(3)}(\lambda)+{\cal O}(s_n),\\
f_{n,\epsilon}(1)=f_n^{(2)}(\lambda\alpha_w R_w(1))+ \alpha_w \widetilde f_{w}(\lambda \alpha_u\alpha_v R_{uv}(1)).
\end{cases}	
\end{align*}
Here clearly appears the recursive construction of our proof that exploits the layered structure of the problem: Theorem \ref{thm:RS_2} 
allows to compute $f_n^{(2)}(\lambda\alpha_w R_w(1))$ (note the ``effective'' signal-to-noise $\lambda\alpha_w R_w(1)$) that we will then use to obtain $f_n^{(3)}(\lambda)$ through the adaptive interpolation method. By the very same steps as in sec.~\ref{sec:dfdt} we get 
\begin{talign}
	&f_{n,\epsilon}'(t)= -\frac{\lambda}{2}\alpha_u\alpha_v\alpha_w\big\{m_{uv}(t)m_w(t)\nonumber\\
&\qquad-\EE\big\langle (Q_uQ_v - m_{uv}(t))(Q_w-m_w(t))\big\rangle_{t,\epsilon}\big\}. \label{40}
\end{talign}
As mentionned in sec. \ref{concentration}, the 
concentration of overlaps Lemma \ref{concentration} generalizes to the present setting. 
Plugging the values of $f_{n,\epsilon}(t)$ at $t=0,1$ and \eqref{40} in the fundamental theorem of calculus and then using the concentration of $Q_u$, $Q_{v}$, $Q_w$ combined with Cauchy-Schwarz then yields
\begin{talign*}
&f_n^{(3)}(\lambda)= \smallO_n(1)+\frac{1}{s_n^2}\!\int d\epsilon \big[f_n^{(2)}(\lambda\alpha_w \int_0^1m_w(t)dt)\\
&+ \alpha_w \widetilde f_{w}(\lambda \alpha_u\alpha_v \int_0^1 m_{uv}(t)dt)\nonumber\\
&+\frac{\lambda }{2}\alpha_u\alpha_v\alpha_w\big\{\int_0^1  m_{uv}(t)m_w(t)dt\nonumber\\
&-{\textstyle \int_0^1 dt(\EE\langle Q_u\rangle_{t,\epsilon}\EE\langle Q_v\rangle_{t,\epsilon} \!-\! m_{uv}(t))(\EE\langle Q_w\rangle_{t,\epsilon}\!-\!m_w)\big\}}\big]. \nonumber
\end{talign*}
We again used the continuity and boundedness of $\widetilde f_{u/v/w}$. {\it Combining this identity with Theorem \ref{thm:RS_2}} (and using Cauchy-Schwarz and the boundedness of the potentials) leads to:
\begin{lemma}[Sum rule]\label{sumRule_2} 
Assume that $\epsilon=(\epsilon_w,\epsilon_{uv})\mapsto R(t,\epsilon)=(R_w(t,\epsilon),R_{uv}(t,\epsilon))$ satisfies the hypotheses of Lemma~\ref{concentration}, and choose $s_n\to 0_+$ according to this lemma. Assume that for all $t \in [0,1]$ and $\epsilon \in [s_n,2s_n]^2$ we have $m_{uv}(t,\epsilon) = \mathbb{E}\langle Q_u \rangle_{t,\epsilon}\mathbb{E}\langle Q_v \rangle_{t,\epsilon}$. Then:
\begin{talign}
&f_n^{(3)}(\lambda)= \frac{1}{s_n^2}\!\int d\epsilon\big[\frac{\lambda }{2}\alpha_u\alpha_v\alpha_w\int_0^1  m_{uv}(t,\epsilon)m_w(t,\epsilon)dt\label{sumRule_form2}\\
&\qquad+ \alpha_w \widetilde f_{w}(\lambda \alpha_u\alpha_v \int_0^1 m_{uv}(t,\epsilon)dt)\nonumber\\
&\qquad+{\adjustlimits \inf_{m_u}\sup_{m_v} }\, f_{\rm pot}^{(2)}(m_u,m_v;\lambda\alpha_w \int_0^1m_w(t,\epsilon)dt)\big]\!+\! \smallO_n(1)\nonumber
\end{talign}	
where $\smallO_n(1)$ is uniform in $t, m_{w},m_{uv},\epsilon$.
\end{lemma}

$\bullet$ {\it Upper bound:} Set $m_w(t)=m_w\in[0,\rho_w]$, and then $R = (R_w,R_{uv})$ as the solution $R(t,\epsilon)=(\epsilon_w+m_w t,\epsilon_{uv}+\int_0^t m_{uv}(s,\epsilon)ds)$ to the ODE: $\partial_t R_w(t) = m_w$, $\partial_t R_{uv}(t) = F_{uv}(t,R(t))\equiv\mathbb{E}\langle Q_u \rangle_{t,\epsilon}\mathbb{E}\langle Q_v \rangle_{t,\epsilon}$ (with values in $[0,\rho_u\rho_v]$), and $R(0) = \epsilon$. By the Cauchy-Lipschitz theorem $R(t,\epsilon)$ is unique and ${\cal C}^1$ (in both arguments). Liouville's formula for the Jacobian determinant of $\epsilon \mapsto R(t,
\epsilon)$ then implies $J_{n,\epsilon}(t) = \exp \{\int_0^t\partial_{R_{uv}} F_{uv}(s,R(s,\epsilon))ds\} \geq 1$ as $\partial_{R_{uv}} F_{uv}\ge 0$. By the same arguments as before $\epsilon \mapsto R(t,
\epsilon)$ is a ${\cal C}^{1}$ diffeomorphism. Recalling \eqref{f3aux}, Lemma.~\ref{sumRule_2} gives: 
\begin{talign*}
f_n^{(3)}(\lambda) =  \frac{1}{s_n^{2}}  \int d\epsilon \, f_{\rm aux}^{(3)}(m_w,\int_0^1 m_{uv}(t,\epsilon)dt;\lambda)+ \smallO_n(1).
\end{talign*}
Thus $\!\limsup_{n\!\to\infty}\!f_n^{(3)}\!(\lambda)\!\le\! { \inf_{m_w}\sup_{m_{uv}}}f_{\rm aux}^{(3)}(m_w,m_{uv};\lambda)$.

$\bullet$ {\it Lower bound:} Fix $R$ as the solution $R(t,\epsilon)=(\epsilon_w+\int_0^t m_w(s,\epsilon)ds,\epsilon_{uv}+\int_0^t m_{uv}(s,\epsilon)ds)$ to the ODE (recall \eqref{MFf_tensor}): $\partial_t R_w(t) = F_w(t,R(t))\equiv -2\widetilde f_w'(\lambda\alpha_u\alpha_v\mathbb{E}\langle Q_u\rangle_{t,\epsilon}\mathbb{E}\langle Q_v\rangle_{t,\epsilon})$ and $\partial_t R_{uv}(t) = F_{uv}(t,R(t)) \equiv \mathbb{E}\langle Q_u\rangle_{t,\epsilon}\mathbb{E}\langle Q_v\rangle_{t,\epsilon}$ with $R(0) = \epsilon$. The solutions verify $m_w(t,\epsilon) \in [0,\rho_w]$ and $m_{uv}(t,\epsilon)=\mathbb{E}\langle Q_u\rangle_{t,\epsilon}\mathbb{E}\langle Q_v\rangle_{t,\epsilon} \in [0,\rho_u\rho_v]$. As previously the Cauchy-Lipschitz theorem implies that $R(t,\epsilon)$ is a $\mathcal{C}^1$ function of both $t$ and $\epsilon$. The Liouville formula for the Jacobian determinant of the map $\epsilon \mapsto R(t,\epsilon)$ gives $J_{n,\epsilon}(t) = \exp\{\int_0^t\partial_{R_w} F_{w}(s,R(s,\epsilon))ds+\int_0^t\partial_{R_{uv}} F_{uv}(s,R(s,\epsilon))ds\} \geq 1$ by non-negativity of the partials. Again, we also have that the map $\epsilon \mapsto R(t,\epsilon)$ a ${\cal C}^1$ diffeomorphism. Lemma.~\ref{sumRule} then leads to formula \eqref{sumRule_form2} with $m_w$ and $m_{uv}$ solutions of the ODE above. Both $\widetilde f_{w}$ and $\inf_{m_u}\sup_{m_v} f_{\rm pot}^{(2)}(m_u,m_v;\,\cdot\,)$ are concave; the latter is a consequence of Theorem~\ref{thm:RS_2} combined with the concavity of $f_n^{(2)}(\lambda)$ (itself concave by $\lambda$-concavity of the mutual information for Gaussian channels and recalling \eqref{2MI}). Jensen's inequality thus yields 
\begin{eqnarray*}
\textstyle f_n^{(3)}(\lambda)\geq \frac{1}{s_n^{2}} \!\int d\epsilon \int_0^1  dt \, f_{\rm aux}^{(3)}(m_w(t,\epsilon),m_{uv}(t,\epsilon);\lambda) \!+\! \smallO_n(1).
\end{eqnarray*}
The same mechanism as in the previous lower bound (for the $p=2$ case) takes place here: 
\begin{align*}
f_{\rm aux}^{(3)}(m_w(t,\epsilon),m_{uv}(t,\epsilon);\lambda)=\sup_{m_{uv}}f_{\rm aux}^{(3)}(m_w(t,\epsilon),m_{uv};\lambda).	
\end{align*}
Indeed, $g_{m_w} : m_{uv} \mapsto f_{\rm aux}^{(3)}(m_w,m_{uv};\lambda)$ is concave with derivative $g_{m_w}'(m_{uv}) =\frac{\lambda}{2}\alpha_u\alpha_v\alpha_w[m_w + 2\widetilde f_w'(\lambda \alpha_u\alpha_v m_{uv})]$. The solutions then verify $g_{m_w(t,\epsilon)}'(m_{uv}(t,\epsilon)) = 0$, so by concavity $g_{m_w(t,\epsilon)}$ reaches its maximum at $m_{uv}(t,\epsilon)$. Therefore,
\begin{talign*}
f_n^{(3)}(\lambda)&\geq \frac{1}{s_n^{2}}\! \int d\epsilon \int_0^1  dt\,  \underset{m_{uv}}{\sup}\, f_{\rm aux}^{(3)}(m_w(t,\epsilon),m_{uv};\lambda) + \smallO_n(1)\\
&\geq {\adjustlimits\inf_{m_w}  \sup_{m_{uv}}}\, f_{\rm aux}^{(3)}(m_w,m_{uv};\lambda) + \smallO_n(1).
\end{talign*}
Taking the $\liminf_{n\to\infty}$ end the proof of the bound, and thus of Theorem \ref{thm:RS_3}. $\QEDA$

\appendix 
 \section{Some sup-inf formulas} \label{appendix_sup_inf}

 This appendix gathers some technical results regarding the manipulation of ``sup-inf'' expressions. The first lemma comes from \cite{barbier2017phase} (Appendix D).
 \begin{lemma} \label{lem:sup_inf}
	 Let $f$ and $g$ be two convex, non-decreasing Lipschitz functions on $\R_+$. 
	 Suppose that $g$ is strictly convex and differentiable.
	 For $q_1, q_2 \!\in\! \R_+$ we define $\psi(q_1,q_2) \!= \!f(q_1) \!+\! g(q_2) \!-\! q_1 q_2$.
	 Then
	 \begin{equation}\label{eq:sup_inf1}
		 \sup_{q_1 \geq 0} \inf_{q_2 \geq 0} \psi(q_1,q_2)
		 =
		 \sup_{q_2 \geq 0} \inf_{q_1 \geq 0} \psi(q_1,q_2)
		 =\!\!\!
		 \sup_{\substack{q_1 = g'(q_2)\\ q_2 = f'(q_1^+)}}\!\!\!
		 \psi(q_1,q_2).
	 \end{equation}
	 Moreover, the above extremas are achieved precisely on the same couples $(q_1,q_2)$ and $f$ is differentiable at $q_1$.
 \end{lemma}

 \begin{lemma}\label{lem:deriv_sup_inf}
	 Let $f$ and $g$ be two convex, non-decreasing Lipschitz functions on $\R_+$. Suppose that $f$ and $g$ are differentiable and strictly convex. Then the function
	 \begin{equation}\label{eq:def_phi}
		 \varphi: t \geq 0 \mapsto 
		 \sup_{q_1 \geq 0} \inf_{q_2 \geq 0} f(tq_1) + g(t q_2) - t q_1 q_2
	 \end{equation}
	 is convex, Lipschitz and non-decreasing. Moreover
	 $\varphi'(0^+) = f'(0^+) g'(0^+)$ and for all $t > 0$:
	 
	 	\vspace{-7pt}
	 	\begin{small}
		 \begin{align*}
			 \varphi'(t^-) &= \min \{ q_1^*(t) q_2^*(t) \, | \, (q_1^*(t),q_2^*(t)) \ \text{optimal couple in \eqref{eq:def_phi}} \},
			 \\
			 \varphi'(t^+) &= \max \{ q_1^*(t) q_2^*(t) \, | \, (q_1^*(t),q_2^*(t)) \ \text{optimal couple in \eqref{eq:def_phi}} \}.
		 \end{align*}
	 	\end{small}
 \end{lemma}

 \begin{proof}
	 Let
	$g^*: x \in  \R \mapsto \sup_{y \in \R_+} \{ xy - g(y) \} \in \R \cup \{+ \infty\}$
	be the Fenchel-Legendre transform of $g$.
	 For $t \geq 0$
	 \begin{equation} \label{eq:phi_sup}
		 \varphi(t)= \sup_{q_1 \geq 0} f(tq_1) - g^*(q_1)
	 \end{equation}
	 (this is true for $t\!>\!0$ and one can verify easily that it is also true for $t\!=\!0$ because $g^*$ is non-decreasing).
	 $\varphi$ is thus a suppremum of convex functions and is therefore convex. 
	 
	 Let $0\!<\!a\!<\!b$.
	 For all $t \!\in \![a,b]$, Lemma \ref{lem:sup_inf} gives that the supremum \eqref{eq:phi_sup} is achieved on a compact set (that does not depend on $t$, but only on $a,b$). Thus Corollary 4 from \cite{milgrom2002envelope} gives that 
	 \begin{align*}
		 \varphi'(t^-) &= \min \{ q_1^*(t) f'(tq_1^*(t)) \, | \, q_1^*(t) \ \text{optimal in \eqref{eq:phi_sup}} \},
		 \\
		 \varphi'(t^+) &= \max \{ q_1^*(t) f'(tq_1^*(t)) \, | \, q_1^*(t) \ \text{optimal in \eqref{eq:phi_sup}} \}.
	 \end{align*}
	 Using Lemma \ref{lem:sup_inf} one see that $f'(t q_1^*(t))$ is equal to the $q_2^*(t)$ from the proposition. $\varphi'(0^+)$ is computed analogously.
	 
	 The fact that $\varphi$ is Lipschitz and non-decreasing follows from the expression of its left- and right-derivatives. Indeed, we know by Lemma \ref{lem:sup_inf} that the optimal couples on \eqref{eq:def_phi} are in $[0,\sup_{x \geq 0} g'(x)] \times [0, \sup_{x \geq 0} f'(x^+)]$.
 \end{proof}
\vspace{4pt}
 \begin{lemma} \label{lem:SE_equiv}
	 Let $f_1, f_2, f_3$ be 3 strictly convex, non-decreasing, differentiable, Lipschitz functions from $\R_+$ to $\R$. Then
	 
	 	\vspace{-15pt}
	 	\begin{small}
		 \begin{align*}
			 &\sup_{q_3 \geq 0} \inf_{r \geq 0} f_3(r) + \sup_{q_1 \geq 0} \inf_{q_2\geq 0} \left\{ f_1(q_2 q_3) + f_2(q_1 q_3) - q_1 q_2 q_3 \right\} - r q_3
			 \\
			 &=
			 \sup_{\substack{q_1 = f'_1(q_2 q_3) \\ q_2 = f'_2(q_1q_3) \\ q_3 = f'_3(q_1 q_2)}} f_1(q_2 q_3) + f_2(q_1 q_3) + f_3(q_1 q_2) - 2 q_1 q_2 q_3.
		 \end{align*}
	 \end{small}
 \end{lemma}
 \vspace{4pt}
 \begin{proof}
	 Let us define
	 $$
	 \varphi: q_3 \geq 0 \mapsto \sup_{q_1 \geq 0} \inf_{q_2\geq 0} \left\{ f_1(q_2 q_3) + f_2(q_1 q_3) - q_1 q_2 q_3 \right\} .
	 $$
	 We know by Lemma \ref{lem:deriv_sup_inf} that $\varphi$ is convex, Lipschitz and non-decreasing over $\R_+$.

	 We will first prove that in the setting of Lemma \ref{lem:sup_inf}, all the quantities of \eqref{eq:sup_inf1} are equal to $\sup_{q_1=g'(q_2)} \psi(q_1,q_2)$.
	 Obviously,
	 \begin{equation}\label{eq:first_ineq}
\sup_{q_1=g'(q_2)} \psi(q_1,q_2) \geq \sup_{\substack{q_1 = g'(q_2)\\ q_2 = f'(q_1^+)}}\!\!\!  \psi(q_1,q_2) \,.
	 \end{equation}
	 Now, let $q_1,q_2 \!\geq\! 0$ such that $q_1 \!=\! g'(q_2)$. The function $r \mapsto \psi(q_1,r)$ is convex and its derivative at $r\!=\!q_2$ vanishes. Thus
	 $$
	 \psi(q_1,q_2) = \inf_{r \geq 0} \psi(q_1,r) \leq \sup_{r_1 \geq 0} \inf_{r_2 \geq 0} \psi(r_1,r_2)
	 $$
	 which combined with \eqref{eq:first_ineq} and \eqref{eq:sup_inf1} gives that $ \sup_{q_1=g'(q_2)} \psi(q_1,q_2)$ is equal to \eqref{eq:sup_inf1}. We now apply this result twice to obtain
	 
	 	\vspace{-7pt}
	 	\begin{small}
		 \begin{align*}
			 &\sup_{q_3 \geq 0} \inf_{r \geq 0} f_3(r) \!+\! \sup_{q_1 \geq 0} \inf_{q_2\geq 0} \left\{ f_1(q_2 q_3) \!+\! f_2(q_1 q_3) \!-\! q_1 q_2 q_3 \right\} \!- \!r q_3
			 \\
			 &=
			 \!\!\!\!\sup_{q_3 = f_3'(r)} \sup_{q_1 = f_1'(q_2 q_3)}\!\!\!\! \left\{ f_3(r) \!+\!  f_1(q_2 q_3) \!+\! f_2(q_1 q_3) \!-\! q_1 q_2 q_3 \!-\! r q_3\right\} .
		 \end{align*}
	 \end{small}
	 Let us add two more constraints to the last supremums, namely ``$r\!=\!q_1 q_2$'' and ``$q_2\!=\!f_2'(q_1 q_2)$''. Adding constraints to a supremum cannot increase it, therefore
	 
	 	\vspace{-7pt}
	 	\begin{small}
		 \begin{align*}
			 &\sup_{q_3 \geq 0} \inf_{r \geq 0} f_3(r) + \sup_{q_1 \geq 0} \inf_{q_2\geq 0} \left\{ f_1(q_2 q_3) + f_2(q_1 q_3) - q_1 q_2 q_3 \right\} - r q_3
			 \\
			 &\geq
			 \!\!\! \sup_{\substack{q_3 = f_3'(r), \  q_1 = f_1'(q_2 q_3) \\ r = q_1 q_2, \ q_2 = f_2'(q_1q_3)}} \!\!\!\!\!\!\!\!\!\left\{ f_3(r) +  f_1(q_2 q_3) + f_2(q_1 q_3) - q_1 q_2 q_3 - r q_3\right\} 
			 \\
			 &=
			 \sup_{\substack{q_1 = f'_1(q_2 q_3) \\ q_2 = f'_2(q_1q_3) \\ q_3 = f'_3(q_1 q_2)}} f_1(q_2 q_3) + f_2(q_1 q_3) + f_3(q_1 q_2) - 2 q_1 q_2 q_3 .
		 \end{align*}
	 \end{small}
	 
	Let us now prove the converse bound. We apply Lemma \ref{lem:sup_inf} twice to obtain
	
		\vspace{-12pt}
	 	\begin{small}
		\begin{align*}
			&\sup_{q_3 \geq 0} \inf_{r \geq 0} f_3(r) + \sup_{q_1 \geq 0} \inf_{q_2\geq 0} \left\{ f_1(q_2 q_3) + f_2(q_1 q_3) - q_1 q_2 q_3 \right\} - r q_3
			\\
			&=
			\!\!\!\!\sup_{\substack{q_3 = f_3'(r) \\ r = \varphi'(q_3^+)}} f_3(r) + \!\!\!\sup_{\substack{q_1 = f_1'(q_2 q_3) \\ q_2 = f_2'(q_1q_3)}} \!\!\!\!\!\left\{ f_1(q_2 q_3) + f_2(q_1 q_3) - q_1 q_2 q_3 \right\} - r q_3
			\\
			&=
			\!\!\!\!\sup_{\substack{q_3 = f_3'(r) \\ r = \varphi'(q_3^+)}} \sup_{\substack{q_1 = f_1'(q_2 q_3) \\ q_2 = f_2'(q_1q_3)}} \!\!\!\!\!\!\!\!\left\{ f_3(r) +  f_1(q_2 q_3) + f_2(q_1 q_3) - q_1 q_2 q_3 - r q_3\right\} .
		\end{align*}
	\end{small}
	 Let now $(r^*,q^*_3)$ be a couple that achieves the first supremum (we know from Lemma \ref{lem:sup_inf} that such a couple exists). 
	 Let now $(q^*_1,q^*_2)$ be a couple that achieves the corresponding second supremum, for which the product $q^*_1 q^*_2$ is maximal.
	 By Lemma \ref{lem:deriv_sup_inf}, we have
	 $r^* = \varphi'(q_3^{*+}) = q^*_1 q^*_2$.
	 However, we know from Lemma \ref{lem:sup_inf} that this couple verifies $q^*_1 = f'_1(q^*_2 q^*_3)$ and $q^*_2 = f'_2(q^*_1 q^*_3)$. Thus
	 
	 \vspace{-10pt}
	 	\begin{small}
	 \begin{align*}
		 &\sup_{q_3 \geq 0} \inf_{r \geq 0} f_3(r)  + \sup_{q_1 \geq 0} \inf_{q_2\geq 0} \left\{ f_1(q_2 q_3) + f_2(q_1 q_3) - q_1 q_2 q_3 \right\} - r q_3
		 \\
		 &=
		 f_1(q_2^* q_3^*) + f_2(q_1^* q_3^*) + f_3(r^*) - q_1^* q_2^* q_3^* - r^* q_3^*
		 \\
		 &=
		 f_1(q_2^* q_3^*) + f_2(q_1^* q_3^*) + f_3(q_1^* q_2^*) - 2 q_1^* q_2^* q_3^*
		 \\
		 &\leq
		 \sup_{\substack{q_1 = f'_1(q_2 q_3) \\ q_2 = f'_2(q_1q_3) \\ q_3 = f'_3(q_1 q_2)}} f_1(q_2 q_3) + f_2(q_1 q_3) + f_3(q_1 q_2) - 2 q_1 q_2 q_3
	 \end{align*}
	 \end{small}
	 which concludes the proof.
 \end{proof}

\section*{Acknowledgments}
J.B acknowledges funding from SNSF grant 200021-156672. Part of this work was done while L.M visited EPFL.
\bibliographystyle{IEEEtran}
\bibliography{refs}

\begin{thebibliography}{10}
\providecommand{\url}[1]{#1}
\csname url@samestyle\endcsname
\providecommand{\newblock}{\relax}
\providecommand{\bibinfo}[2]{#2}
\providecommand{\BIBentrySTDinterwordspacing}{\spaceskip=0pt\relax}
\providecommand{\BIBentryALTinterwordstretchfactor}{4}
\providecommand{\BIBentryALTinterwordspacing}{\spaceskip=\fontdimen2\font plus
\BIBentryALTinterwordstretchfactor\fontdimen3\font minus
  \fontdimen4\font\relax}
\providecommand{\BIBforeignlanguage}[2]{{%
\expandafter\ifx\csname l@#1\endcsname\relax
\typeout{** WARNING: IEEEtran.bst: No hyphenation pattern has been}%
\typeout{** loaded for the language `#1'. Using the pattern for}%
\typeout{** the default language instead.}%
\else
\language=\csname l@#1\endcsname
\fi
#2}}
\providecommand{\BIBdecl}{\relax}
\BIBdecl

\bibitem{BarbierM17a}
\BIBentryALTinterwordspacing
J.~Barbier and N.~Macris, ``The adaptive interpolation method: {A} simple
  scheme to prove replica formulas in bayesian inference,'' \emph{Probability
  Theory and Related Fields}, 2018. [Online]. Available:
  \url{http://arxiv.org/abs/1705.02780}
\BIBentrySTDinterwordspacing

\bibitem{sidiropoulos2016}
N.~Sidiropoulos, L.~De~Lathauwer, X.~Fu, K.~Huang, E.~Papalexakis, and
  F.~Christos, ``Tensor decomposition for signal processing and machine
  learning,'' \emph{arXiv preprint arXiv:1607.01668v2}, 2016.

\bibitem{cichoki2015}
A.~Cichocki, D.~Mandic, L.~De~Lathauwer, Q.~Zhou, Q.~Zhao, C.~Caiafa, and
  A.~Phan, ``Tensor decompositions for signal processing applications: From
  two-way to multiway component analysis,'' \emph{Signal Processing Magazine,
  IEEE, vol. 32, no. 2, pp. 145–163}, 2015.

\bibitem{kolda2009}
T.~Kolda and B.~Bader, ``Tensor decompositions and applications,'' \emph{SIAM
  REVIEW vol 51, no 3}, 2009.

\bibitem{mezard1990spin}
M.~M{\'e}zard, G.~Parisi, and M.-A. Virasoro, ``Spin glass theory and beyond.''
  1990.

\bibitem{Donoho10112009}
\BIBentryALTinterwordspacing
D.~L. Donoho, A.~Maleki, and A.~Montanari, ``Message-passing algorithms for
  compressed sensing,'' \emph{Proceedings of the National Academy of Sciences},
  vol. 106, no.~45, pp. 18\,914--18\,919, 2009. [Online]. Available:
  \url{http://www.pnas.org/content/106/45/18914.abstract}
\BIBentrySTDinterwordspacing

\bibitem{BayatiMontanari10}
M.~Bayati and A.~Montanari, ``The dynamics of message passing on dense graphs,
  with applications to compressed sensing,'' \emph{IEEE Trans. on Inf. Theory},
  vol.~57, no.~2, pp. 764 --785, 2011.

\bibitem{2017arXiv170100858L}
T.~{Lesieur}, F.~{Krzakala}, and L.~{Zdeborov{\'a}}, ``{Constrained Low-rank
  Matrix Estimation: Phase Transitions, Approximate Message Passing and
  Applications},'' \emph{ArXiv e-prints}, Jan. 2017.

\bibitem{mezard2009information}
M.~Mezard and A.~Montanari, \emph{Information, physics and computation}.\hskip
  1em plus 0.5em minus 0.4em\relax Oxford University Press, 2009.

\bibitem{Guerra-Toninelli-2002}
F.~Guerra and F.~Toninelli, ``Quadratic replica coupling in the {S}herrington-
  {K}irkpatrick mean field spin glass model,'' \emph{J. Math. Phys.}, vol.~43,
  p. 3704–3716, 2002.

\bibitem{Guerra-2003}
F.~Guerra, ``Replica broken bounds in the mean field spin glass model,''
  \emph{Comm. Math. Phys.}, vol. 233, pp. 1--12, 2003.

\bibitem{Talagrand-annals-2006}
M.~Talagrand, ``The {P}arisi formula,'' \emph{Ann. Math.}, vol. 163, p.
  221–263, 2006.

\bibitem{Parisi-1980}
G.~Parisi, ``A sequence of approximate solutions to the {S}-{K} model for spin
  glasses,'' \emph{J. Phys. A}, vol. 13 L-115, 1980.

\bibitem{Panchenko2013}
D.~Panchenko, \emph{The Sherrington-Kirkpatrick Model}.\hskip 1em plus 0.5em
  minus 0.4em\relax Springer Monographs in Mathematics, 2013.

\bibitem{Giurgiu_SCproof}
A.~Giurgiu, N.~Macris, and R.~Urbanke, ``Spatial coupling as a proof technique
  and three applications,'' \emph{IEEE Trans. on Information Theory}, vol.~62,
  no.~10, pp. 5281--5295, Oct 2016.

\bibitem{barbier_allerton_RLE}
J.~Barbier, M.~Dia, N.~Macris, and F.~Krzakala, ``{The Mutual Information in
  Random Linear Estimation},'' in \emph{in the 54th Annual Allerton Conference
  on Communication, Control, and Computing}, September 2016.

\bibitem{barbier_ieee_replicaCS}
\BIBentryALTinterwordspacing
J.~Barbier, N.~Macris, M.~Dia, and F.~Krzakala, ``{Mutual Information and
  Optimality of Approximate Message-Passing in Random Linear Estimation}.''
  [Online]. Available: \url{https://arxiv.org/pdf/1701.05823v1.pdf}
\BIBentrySTDinterwordspacing

\bibitem{XXT}
J.~Barbier, M.~Dia, N.~Macris, F.~Krzakala, T.~Lesieur, and L.~Zdeborov\'{a},
  ``Mutual information for symmetric rank-one matrix estimation: A proof of the
  replica formula,'' in \emph{Advances in Neural Information Processing Systems
  (NIPS) 29}, 2016, pp. 424--432.

\bibitem{2016arXiv161103888L}
M.~{Lelarge} and L.~{Miolane}, ``{Fundamental limits of symmetric low-rank
  matrix estimation},'' \emph{ArXiv e-prints}, Nov. 2016.

\bibitem{2017arXiv170200473M}
L.~{Miolane}, ``{Fundamental limits of low-rank matrix estimation: The
  non-symmetric case},'' \emph{ArXiv e-prints}, Feb. 2017.

\bibitem{aizenman2003extended}
M.~Aizenman, R.~Sims, and S.~L. Starr, ``Extended variational principle for the
  {S}herrington-{K}irkpatrick spin-glass model,'' \emph{Physical Review B},
  vol.~68, no.~21, p. 214403, 2003.

\bibitem{barbier2017phase}
J.~Barbier, F.~Krzakala, N.~Macris, L.~Miolane, and L.~Zdeborov{\'a}, ``Optimal
  errors and phase transitions in high-dimensional generalized linear models,''
  \emph{arXiv preprint arXiv:1708.03395}, 2017.

\bibitem{krzakala2016mutual}
F.~Krzakala, J.~Xu, and L.~Zdeborov{\'a}, ``Mutual information in rank-one
  matrix estimation,'' \emph{arXiv:1603.08447}, 2016.

\bibitem{lesieur2015mmse}
T.~Lesieur, F.~Krzakala, and L.~Zdeborov\'a, ``Mmse of probabilistic low-rank
  matrix estimation: Universality with respect to the output channel,'' in
  \emph{Annual Allerton Conference}, 2015.

\bibitem{deshpande2015asymptotic}
Y.~Deshpande, E.~Abbe, and A.~Montanari, ``Asymptotic mutual information for
  the two-groups stochastic block model,'' \emph{arXiv:1507.08685}, 2015.

\bibitem{2017arXiv170108010L}
T.~{Lesieur}, L.~{Miolane}, M.~{Lelarge}, F.~{Krzakala}, and
  L.~{Zdeborov{\'a}}, ``{Statistical and computational phase transitions in
  spiked tensor estimation},'' \emph{ArXiv e-prints}, Jan. 2017.

\bibitem{Talagrand2011spina}
M.~Talagrand, \emph{Mean Field Models for Spin Glasses. Volume I: Basic
  Examples}.\hskip 1em plus 0.5em minus 0.4em\relax Springer Verlag, 2011.

\bibitem{bayati2011dynamics}
M.~Bayati and A.~Montanari, ``The dynamics of message passing on dense graphs,
  with applications to compressed sensing,'' \emph{IEEE Trans. on Information
  Theory}, 2011.

\bibitem{Montanari-Javanmard}
A.~Javanmard and A.~Montanari, ``State evolution for general approximate
  message passing algorithms, with applications to spatial coupling,'' \emph{J.
  Infor. \& Inference}, vol.~2, p. 115, 2013.

\bibitem{guerra2002thermodynamic}
F.~Guerra and F.~L. Toninelli, ``The thermodynamic limit in mean field spin
  glass models,'' \emph{Communications in Mathematical Physics}, vol. 230,
  no.~1, pp. 71--79, 2002.

\bibitem{GKSmacris2007}
N.~Macris, ``{G}riffith {K}elly {S}herman correlation inequalities: A useful
  tool in the theory of error correcting codes,'' \emph{IEEE Transactions on
  Information Theory}, vol.~53, no.~2, pp. 664--683, 2007.

\bibitem{KoradaMacris_CDMA}
S.~B. Korada and N.~Macris, ``Tight bounds on the capacity of binary input
  random {CDMA} systems,'' \emph{IEEE Trans. on Information Theory}, vol.~56,
  no.~11, pp. 5590--5613, Nov 2010.

\bibitem{korada2009exact}
------, ``Exact solution of the gauge symmetric p-spin glass model on a
  complete graph,'' \emph{Journal of Statistical Physics}, 2009.

\bibitem{milgrom2002envelope}
P.~Milgrom and I.~Segal, ``Envelope theorems for arbitrary choice sets,''
  \emph{Econometrica}, vol.~70, no.~2, pp. 583--601, 2002.

\end{thebibliography}
\end{document}